\documentclass[10pt,aps,prd,onecolumn,showpacs,amsmath,amssymb,nofootinbib,eqsecnum,preprintnumbers,superscriptaddress]{revtex4-2}

\usepackage{esvect}
\usepackage{xcolor}
\usepackage{amsfonts}
\usepackage{physics}
\usepackage{xspace}
\usepackage{mathtools}
\usepackage{amsthm}
\usepackage{tikz}
\usepackage{comment}
\usepackage[colorlinks=true,
            citecolor=red,
            linkcolor=blue,
            urlcolor=violet,
            filecolor=cyan,
            backref=false]{hyperref}
\usepackage{multirow}
\usepackage{titlesec}

\newtheorem{theorem}{Theorem}
\newtheorem{corollary}{Corollary}
\newtheorem{definition}{Definition}

\titleformat{\paragraph}[block]{\filcenter}{}{0pt}{}

\DeclareMathSymbol{*}{\mathbin}{symbols}{"03} 
\DeclareMathSymbol{\ast}{\mathbin}{symbols}{"03}

\begin{document}


\title{
Intrinsically-defined higher-derivative Carrollian scalar field theories without Ostrogradsky instability
}

\author{Poula Tadros}
\email{poula.tadros@matfyz.cuni.cz}
\affiliation{Institute of Theoretical Physics, Faculty of Mathematics and Physics, Charles University,
V Hole\v{s}ovi\v{c}k\'ach 2, Prague 180 00, Czech Republic}

\author{Ivan Kol\'a\v{r}}
\email{ivan.kolar@matfyz.cuni.cz}
\affiliation{Institute of Theoretical Physics, Faculty of Mathematics and Physics, Charles University,
V Hole\v{s}ovi\v{c}k\'ach 2, Prague 180 00, Czech Republic}

\date{\today}
\begin{abstract}
    We derive the most generic Carrollian higher derivative free scalar field theory intrinsically on a Carrollian manifold. The solutions to these theories are massless free particles propagating with speeds depending on the coupling constants in the Lagrangian, thus, allowing interference solutions which are not allowed on a Lorentzian manifold. This demonstrates that the set of solutions to the Carrollian theories is much larger than that of their Lorentzian counterparts. We also show that Carrollian higher derivative theories are more resistant to Ostrogradsky's instabilities. These instabilities can be resolved by choosing the coupling constants appropriately in the Carrollian Lagrangian, something that was proven to be impossible in Lorentzian theories.
\end{abstract}

\maketitle

\section{Introduction}
The \textit{Carrollian limit} was first considered by Levy-Leblond \cite{Levy-Leblond} and independently by Sen Gupta \cite{Gupta} as an In\"on\"u-Wigner contraction of the Poincar\'e group where the speed of light $c$ goes to zero, ${c\to0}$. However, at the time of this discovery there were no physical applications for such limit. Thus, it was only studied in a mathematical context until recently, when the Carrollian limit was linked to many applications in physics. Now, Carrollian physics and Carrollian structures are studied in the context of Carroll particles \cite{Zhang:2023jbi,Bergshoeff_2014,Marsot:2021tvq,Marsot:2022imf,Kasikci:2023tvs,Casalbuoni:2023bbh,Cerdeira:2023ztm,Kamenshchik:2023kxi}, condensed matter physics \cite{Bagchi:2022eui,Kubakaddi_2021,Kononov_2021,Bidussi:2021nmp,Figueroa-OFarrill:2023vbj,Figueroa-OFarrill:2023qty,Gromov:2022cxa,Pretko:2020cko}, field theory \cite{PhysRevD.106.085004,Chen:2023pqf,Bergshoeff:2022eog,Henneaux:2021yzg,Rivera-Betancour:2022lkc,Bagchi:2022eav}, conformal field theory \cite{Bagchi:2019xfx,Bagchi:2019clu,Bagchi:2021gai,PhysRevD.103.105001,Chen:2023pqf,Banerjee:2020qjj,Bekaert:2022oeh,Baiguera:2022lsw}, fluid mechanics \cite{Bagchi:2023ysc,Ciambelli:2018wre,Ciambelli:2018xat,campoleoni2019two,Ciambelli:2020eba,10.21468/SciPostPhys.9.2.018}, cosmology \cite{deBoer:2021jej,Bonga:2020fhx}, string theory \cite{PhysRevLett.123.111601,Bagchi:2021ban,Cardona:2016ytk,Isberg:1993av,Bagchi:2015nca,Fursaev:2023lxq,Fursaev:2023oep}, gravity \cite{Perez:2021abf,Perez:2022jpr,Hartong:2015xda,Figueroa-OFarrill:2021sxz,Hansen:2021fxi,Gomis:2020wxp,Bergshoeff:2022qkx,Guerrieri:2021cdz,Hansen:2020wqw} (it is regarded as the strong coupling limit of gravity theories \cite{Anderson:2002zn}), black holes \cite{Donnay:2019jiz,Grumiller:2019tyl,Redondo-Yuste:2022czg,Marsot:2022imf,Anabalon:2021wjy,Ecker:2023uwm,Ciambelli:2023tzb,Gray:2022svz,Bicak:2023vxs,Marsot:2022qkx}, null boundaries and flat holography \cite{Herfray:2021qmp,Chandrasekaran:2021hxc,Bagchi:2019clu,Ciambelli:2019lap,Bagchi:2016bcd,Ciambelli:2018wre,Figueroa-OFarrill:2021sxz,Herfray:2021qmp,Donnay:2022aba,Campoleoni:2023fug,Salzer:2023jqv}.

In the context of field theory, Carrollian theories were previously considered as limits of Lorentzian theories. That was until recently where scalar field theory was defined internally on a Carrollian manifold \cite{Ciambelli:2023xqk}. It is evident from both approaches that not all Carrollian theories can be derived as limits of Lorentzian theories. On the other hand, all Carrollian theories resulting from such limiting procedure can be defined internally without referring to a Lorentzian theory. Another discrepancy between the two approaches concerns the Carroll expansion itself. Different orders of the expansion give rise to different Carrollian theories. Thus, the resulting theories can not coexist or interact. That is because different orders require different rescalings by different powers of $c$ to be finite but non-vanishing in the Carroll limit. A mixed action then takes away the freedom to rescale each order separately resulting in only one order surviving \cite{Rivera-Betancour:2022lkc}. However, if the action is defined intrinsically, there would be no reason to disregard such coexistence of these Carrollian theories. Thus, defining Carrollian theories only by a limiting procedure restricts the resulting dynamics for no good reason. Consequently, in such theories, we typically have trivial dynamics for free particles and fields. This was shown to be false in the case of intrinsically defined Carrollian theories \cite{Ciambelli:2023tzb,Ciambelli:2023xqk}. The reason for non-trivial dynamics is precisely the coexistence and mixing between different actions.

In this paper, we define Carrollian higher-derivative scalar field theory intrinsically without referring to a limiting procedure requiring only diffeomorphism invariance but not necessarily Carroll boost invariance. We show explicitly that there are infinitely many such Carrollian theories that are not reachable by taking the limit of the Lorentzian higher derivative scalar field theory, and reinforce the results of \cite{Ciambelli:2023xqk} that these theories have non-trivial dynamics. In fact, their dynamics is far richer than the Lorentzian higher derivative scalar field theory. This rich dynamics is captured by the fact that the solutions to the field equations represent massless particles propagating with different speeds, giving rise to interference patterns impossible in Lorentzian theories. We also discuss the stability of the Carrollian higher derivative scalar field theory compared to its Lorentzian counterpart. Specifically, we focus on \textit{Ostrogradsky's instabilities}, which occur in higher-derivative field theories due to the Hamiltonian becoming unbounded as a consequence of its linear dependence on the momenta associated with the extra degrees of freedom defined due to the presence of higher derivatives in the Lagrangian.
We show that there is a striking difference between Carrollian and Lorentzian theories. It is well known that Ostrogradsky's instabilities in Lorentzian theories can not be removed by a choice of the parameters in the Lagrangian but require constraints to the phase space \cite{Chen:2012au}. In contrast, in Carrollian theories, we can get rid of these instabilities by adequate choices of the coupling constants (which are parameters of the theory) without the need to impose additional constraints. The paper is organised as follows:
\begin{itemize} \itemsep0em
    \item In Sec. \ref{sec2}, we give a brief introduction on Carrollian geometry and the tools we use throughout the paper. 
    \item In Sec. \ref{sec3}, we review Lorentzian higher-derivative scalar field theory on flat spacetime and its plane wave solution. We also give a short review on Ostrogradsky's theorem and Ostrogradsky's instabilities.
    \item In Sec. \ref{sec4}, we intrinsically define the Carrollian higher-derivative scalar field theory with $n$ derivatives. We then inspect its solutions and show that, in general, it represents $n$ massless particles moving with different speeds. 
    \item In Sec. \ref{sec6}, we discuss Ostrogradsky's instabilities for Carrollian higher derivative scalar field theories, and derive the condition to remove them. We then formulate Ostrogradsky's theorem for these theories.
    \item In Sec. \ref{sec7}, we conclude the paper with a summary and possible future directions.
\end{itemize}

\section{Introduction to Carrollian geometry}\label{sec2}

We begin this section by defining a Carrollian manifold as a $(d+1)-$dimensional manifold equipped with a Carrollian structure, i.e., a nowhere vanishing vector field $v^{\mu}$ and a degenerate metric $h_{\mu\nu}$ such that $v^{\mu}h_{\mu\nu}=0$. We then define an associated one form $E_{\mu}$ to $v^{\mu}$ such that $E_{\mu}v^{\mu}=1$ which plays the role of an Ehresmann connection. We can also define the projector $h_{\mu}^{\nu} = \delta_{\mu}^{\nu}-E_{\mu}v^{\nu}$.
We can alternatively define a Carrollian manifold by contracting the light cone at each point of a Lorentzian manifold $\mathcal{M}$. This forces a splitting of the manifold's tangent bundle $T\mathcal{M}$ into a vertical bundle $\mathrm{Ver} \mathcal{M}$ whose fibers are isomorphic to the timelike vector flow, and horizontal bundles $\mathrm{Hor} \mathcal{M}$ whose fibers are isomorphic to spacelike slices of $\mathcal{M}$. Thus, we have $T\mathcal{M} = \mathrm{Ver} \mathcal{M} \oplus \mathrm{Hor} \mathcal{M}$. The result is a foliation of the manifold, by lines corresponding to the flow of $v^{\mu}$, into codimension 1 submanifolds with metric $h_{\mu\nu}$ as leaves. It is worth mentioning that is construction is different from the usual foliation of Lorentzian manifolds into codimension 1 submanifolds (as leaves) without the light cone contraction since the timelike vector flow in the Lorentzian case has the freedom to propagate inside the light cone from one leaf to another.

Given a Carrollian manifold, we can define a covariant derivative compatible with the Carroll structure (generally torsionfull). In general the covariant derivative is defined by the conditions
\begin{equation}
    \begin{aligned}
        \nabla_{\mu}h^{\nu\rho} &= 0, \hspace{10pt} \nabla_{\mu}h_{\nu\rho} = 0,\hspace{10pt}
        (\nabla_{\mu}-\omega_{\mu})v^{\nu}=0, \hspace{10pt}  (\nabla_{\mu}+\omega_{\mu})E_{\nu}=0,
    \end{aligned}
\end{equation}
where $\omega_{\mu} = k E_{\mu}-\pi_{\mu}$ such that $k$ is a constant and $\pi_{\mu}$ is a covector satisfying $\pi_{\mu}v^{\mu}=0$. For the explicit calculation of the torsion and curvature of this covariant derivative, see \cite{Hansen:2021fxi}. We distinguish between two cases: 
\begin{enumerate} 
    \item If $k$ and $\pi_{\mu}$ are non zero, then the structure is called a weak Carrollian structure.
    \item If $k=0$ and $\pi_{\mu}=0$, then it is called a strong Carrollian structure.
\end{enumerate}
In this paper, we assume a strong Carrollian structure, a torsionless connection, and a flat spacetime.

Having defined the full structure, we will distinguish between two distinct types of Carrollian transformations on curved Carrollian manifolds, namely, Carroll boosts and Carroll diffeomorphisms. Local Carroll boosts are defined by
\begin{equation}
    \begin{aligned}
        \delta_{\lambda}v^{\mu} = 0,\hspace{10pt} \delta_{\lambda}h_{\mu\nu} = 0, \hspace{10pt} \delta_{\lambda}h^{\mu\nu} = h^{\mu\rho}\lambda_{\rho}v^{\nu} +  h^{\nu\rho}\lambda_{\rho}v^{\mu}, \hspace{10pt} \delta_{\lambda}E_{\mu}=\lambda_{\mu},
    \end{aligned}
\end{equation}
where $\lambda_{\mu}$ is a covector field on the Carrollian manifold satisfying $v^{\mu}\lambda_{\mu}=0$. In the case where $\lambda_{\mu}$ is a constant covector field, the boosts are called global Carroll boosts. Notice that $v^{\mu}$ and $h_{\mu\nu}$ are boost invariant, hence they are often regarded as the fundamental Carroll objects. This choice is not unique; in fact, it is possible to choose $E_{\mu}$ and $h^{\mu\nu}$ to be the fundamental object; this possibility was reviewed in \cite{Baiguera:2022lsw}. To ensure that a foliation with spatial leaves is well defined, $E_{\mu}$ must satisfy the Frobenius condition $\textbf{E}\wedge d\textbf{E}=0$ where $d\textbf{E}$ is the exterior derivative of the Ehresmann connection. In adapted coordinates $(t,x^a)$ where $a\in \{1,2,\cdots,d\}$, the Ehresmann connection is written as
\begin{equation}\label{ehr}
    E_{\mu}dx^{\mu} = -dt+b_adx^a,
\end{equation}
where $b_a$ is the shift vector. The form \eqref{ehr} does not satisfy the Frobenius condition in general, thus, we have to choose an appropriate boost frame so that the Frobenius condition is satisfied. For simplicity, we choose the frame where $b_a = 0$. In this boost frame, the Carrollian structure in flat space in adapted coordinates can be written as
\begin{equation}
\begin{aligned}
    v^{\mu}\partial_{\mu}=\partial_t, \hspace{10pt} E_{\mu}dx^{\mu} = -dt, \hspace{10pt} h_{\mu\nu}dx^{\mu}dx^{\nu}=\delta_{ab}dx^{a}dx^{b}, \hspace{10pt} 
h^{\mu\nu}\partial_{\mu}\partial_{\nu}=\delta^{ab}\partial_{a}\partial_{b}.
   \end{aligned} 
\end{equation}

Carroll diffeomorphisms on the other hand are restricted type of diffeomorphisms (compared to the Lorentzian case) defined by
\begin{equation}
    t\to t'(t,x), \hspace{2cm} x \to x'(x).
\end{equation}
The Jacobian of these transformations is given by
\begin{equation}
\begin{pmatrix}
    j(t,x) & j_a(t,x) \\ 0 & j^a_b(x)
\end{pmatrix},
\end{equation}
where $j(t,x) = \tfrac{\partial t'}{\partial t}, j_a(t,x) = 
\tfrac{\partial t'}{\partial x^a}$, and $j^a_b(x) = 
\tfrac{\partial x'^{a}}{\partial x^b}$ \cite{Ciambelli:2018xat}. Tensors are then transformed according to the position of their indices either by the Jacobian or by its inverse.

In this paper, we consider scalar field Lagrangians on Carrollian manifolds (with strong Carrollian structure) that are diffeomorphism invariant but not necessarily Carroll boost invariant. The main reason is that magnetic or mixed actions can not be boost invariant since the fields couple to $h^{\mu\nu}$ which is not boost invariant. These theories are called Carrollian in the sense of \cite{Rivera-Betancour:2022lkc} where the scalar field is defined on a Carrollian background and only gaining boost invariance if an extra term is added (a conformal coupling term), this notion was also used in \cite{Ciambelli:2023xqk} where dynamics were derived from scalar fields on a Carroll manifold without requiring boost invariance. Notice that in other papers, for example in \cite{Baiguera:2022lsw}, the main attribute of a Carrollian theory is boost invariance, so by this definition the theory considered here is not Carrollian, however, we will not use this nomenclature in this paper.

\section{Lorentzian higher derivative scalar field theory}\label{sec3}

\subsection{Lagrangian and field equations}
The Lagrangian for a relativistic higher derivative massless scalar field theory is given by
\begin{equation}\label{rel lag}
    \mathcal{L} = \phi \Box^n \phi,
\end{equation}
where $\phi$ is the scalar field and $\Box = g^{\mu\nu}\nabla_{\mu}\nabla_{\nu}$. The field equations of such a theory is given by
\begin{equation}\label{EOM}
    \Box^n \phi = 0.
\end{equation}
A remark that would be important later is that any solution to the usual wave equation ($\Box \phi = 0$) is also a solution for \eqref{EOM}.  In a flat space, a solution for \eqref{EOM} is a plane wave
\begin{equation}\label{sol}
    \phi = A e^{i(\omega t- \vec{k}.\vec{x})},
\end{equation}
 where $A$ is the amplitude, $\omega$ is the frequency and $k$ is the wave number. The dispersion relation for this wave is given by 
 \begin{equation}
     \omega = \pm c k,
 \end{equation}
where $k = |\vec{k}|$. The general solution is a superposition of such waves with positive and negative frequencies (energies). An important remark is that in this case all waves propagate at the speed of light, i.e. the theory does not allow them to have different propagation speeds. 

\subsection{Ostrogradsky's theorem}
In this subsection we review the Ostrogradsky's theorem \cite{Ostrogradsky:1850fid}, which concerns the stability of theories with higher time derivatives, in the case of scalar field theories whose Lagrangians contain only one scalar field.

\begin{definition}
 Let $\mathcal{L}(\phi, \partial_{\mu}\phi, \partial_{\mu}^2\phi,..., \partial^n_{\mu}\phi)$ be a higher derivative Lagrangian. $\mathcal{L}$ is called non-degenerate if \begin{equation}
    \tfrac{\partial \mathcal{L}}{\partial \phi^{(n)}} \neq 0,
\end{equation}
where $\phi^{(n)}$ is the $n-$th time derivative of the field $\phi$.
\end{definition}
\begin{theorem}[Ostrogradsky, 1850] \cite{Ostrogradsky:1850fid,Chen:2012au}
 Any non-degenerate higher time derivative Lagrangian leads to a Hamiltonian that is not bounded.
\end{theorem}

\noindent This means that the vacuum of the field theory is unstable. The proof for scalar field theories with one real scalar is given by constructing the Hamiltonian as follows:
First of all, due to non-degeneracy, the higher order time derivative can be expressed in terms of lower order derivatives
\begin{equation}
    \phi^{(n)} = f(\phi,\phi^{(1)},\cdots,\phi^{(n-1)}).
\end{equation}
The canonical coordinates and momenta are defined as
\begin{equation}\label{can mom}
    \begin{aligned}
        Q_i = \phi^{(i-1)}, \hspace{1cm} P_i=\sum_{j=i}^{n}\big( - \tfrac{d}{dt}\big)^{j-i} \tfrac{\partial \mathcal{L}}{\partial \phi^{(j)}},
    \end{aligned}
\end{equation}
where $i\in \{ 1,2,\cdots, n \}$. Thus, the phase space coordinates are
$
    Q_0,Q_1,\cdots,Q_{n-1},P_1,P_2,\cdots,P_n.
$
The Hamiltonian is given by
\begin{equation}
    \mathcal{H} = \sum_{i=0}^n P_{i+1}\phi^{(i)}-\mathcal{L} = \sum_{i=1}^{n-1}P_iQ_{i} + P_nf(Q_1,\cdots, Q_{n-1},P_n) - \mathcal{L},
\end{equation}
which is unbounded due to the linearity in $P_1, \cdots,P_{n-1}$ but if the function $f$ is of a suitable form, the Hamiltonian will not be linear in $P_{n}$. This is called the Ostrogradsky's instability since the Hamiltonian, in general, is not bounded from below. Thus, in the presence of interactions the vacuum of the field theory will decay indefinitely to infinite negative energies. It is clear from this construction that if the Lagrangian contains up to a first-order derivative of the field, $f$ is computed to eliminate linearity in the Hamiltonian. Higher-derivative field theories have been constructed with degenerate Lagrangians to circumvent this instability. However, it is not possible by choosing specific values for the parameters in the Lagrangian but by highly restricting the phase space (by the introduction of constraints), i.e., dynamics and possible solutions of such theories must be heavily restricted. In what follows, we show that Carrollian theories do not suffer from this problem. Stable higher-derivative Carrollian theories, on the other hand, can be systematically defined without the need of constraints, i.e., having less restricted dynamics than Lorentzian theories. For more discussion on Ostrogradsky instability, possible generalizations and quantum considerations see \cite{Woodard:2015zca,Motohashi:2020psc} and the references therein.

\section{Intrinsic Carroll higher derivative scalar field theory}\label{sec4}
In this section, we define higher-derivative scalar field theory on a flat Carrollian manifold and demonstrate that the space of solutions is strictly larger than its counterpart on a Lorentzian manifold. The only diffeomorphism invariant objects we have on a Carrollian manifold are $v^{\mu}\partial_{\mu}\phi$ and $h^{\mu\nu}\partial_{\mu}\partial_{\nu}\phi$ from which we construct a diffeomorphism invariant Lagrangian for higher-derivative scalar field theory. Restricting ourselves to only terms with $2n-$th derivatives, the list of diffeomorphism invariant terms can be parameterized using the number of $v^{\mu}\partial_{\mu}\phi$ objects involved, i.e.,
\begin{equation}\label{lagr22}
    \mathcal{L}_{2i} = \phi (v^{\mu}\partial_{\mu})^{2i}(h^{\alpha\beta}\partial_{\alpha}\partial_{\beta})^{n-i}\phi,
\end{equation}
where $i \in \{ 0,1, \cdots, n\}$, and the most general Lagrangian will be a linear combination of these terms:
\begin{equation}\label{lagr}
    \mathcal{L} = \sum_{i=0}^n \Bar{g}_i \mathcal{L}_{2i},
\end{equation}
where $\Bar{g}_i$ are arbitrary coupling constants. Notice that the Lagrangian is not Carroll boost invariant but only diffeomorphism invariant (by construction), this means that non zero local energy flux is not forbidden and we can get non trivial dynamics. From now on, we will use adapted coordinates, i.e., $v^{\mu}\partial_{\mu}=\partial_t$ and $h^{\alpha\beta}\partial_{\alpha}\partial_{\beta} = \partial_a\partial^a$, where $a \in \{1,2,\cdots,d \}$, is the spatial Laplacian operator.

Since the coupling constants $\Bar{g}_i$ are arbitrary, without loss of generality, we can choose them to be
\begin{equation}
    \Bar{g}_{i}= \sum_{\mathclap{\substack{k_1,k_2,\cdots,k_i =1 \\ k_1>k_2>\cdots>k_n}}}^n g_1g_2 \cdots \hat{g}_{k_1}\cdots \hat{g}_{k_2}\cdots\hat{g}_{k_{n-i}}\cdots g_n,
\end{equation}
where $\hat{g}_{k}$ means that $g_{k}$ is being omitted from the product. Substituting in the Lagrangian \eqref{lagr}, rearranging the terms and factorizing we get
\begin{equation}\label{newlag}
    \mathcal{L}= \phi \Box_{g_1}\Box_{g_2}\cdots\Box_{g_n}\phi,
\end{equation}
where $\Box_{g_k}= g_k \partial_t^2 + \partial_i\partial^i$, and $g_k$ are taken to be negative. This is a large family of Lagrangians parameterized by the coupling constants. The Lorentzian theory is recovered when setting $g_1=g_2=\cdots=g_n=-1/c^2$, which clearly shows that Carrollian theories are more general and should not be thought of only as limits of Lorentzian theories.

The equations of motion derived by varying \eqref{newlag}, using the same procedure as in \cite{Ciambelli:2023xqk} is
\begin{equation}\label{eom1}
    \Box_{g_1}\Box_{g_2}\cdots\Box_{g_n}\phi=0.
\end{equation}
We are interested in solutions to this equation of the form \eqref{sol} (plane waves). Substituting \eqref{sol} into \eqref{eom1} and noting that all the differential operators $\Box_{g_k}$ commute, we get $n$ different dispersion relations of the form
\begin{equation}
    \omega = \pm \tfrac{1}{\sqrt{- g_i}} k,
\end{equation}
where $i \in \{ 1,2,\cdots,n \}$. This means that we have coexisting $n$ massless particles, each of which moves in vacuum with different speeds (not necessarily the speed of light). This can be thought of as a generalization of Lorentzian theories where any massless particle must move at the speed of light. We can still recover the solution of the Lorentzian theory by setting all coupling constants to $-1/c^2$, which is also the only solution we get if we begin with the higher derivative Lorentzian theory and then take the Carroll limit as we will show in the next section. Another important aspect to stress on is that, as plane waves, the Carrollian theory permits interference patterns that were forbidden in any Lorentzian theory. Namely, interference between multiple waves moving with different speeds, this is a direct consequence of the superposition principle for equation \eqref{eom1}. This stresses even further that Lorentzian theories have more restrictive dynamics than Carrollian theories.

Another equivalent way to derive the above results is by considering the Green's function for the differential operator $\Box_{g_1}\Box_{g_2}\cdots\Box_{g_n}$ defined by
\begin{equation}\label{green}
    \Box_{g_1}\Box_{g_2}\cdots\Box_{g_n}G(x-x')=-i\delta^{(d+1)}(x-x'),
\end{equation}
where $x = (t,\vec{x})$ and $x'=(t',\vec{x}')$ are position vectors on the Carrollian manifold, and $\delta^{(d+1)}(x-x')$ is the covariant delta distribution on the manifold defined by
\begin{equation}
    \int_{\mathcal{C}} \epsilon  \ \delta^{(d+1)}(x-x') f(x) = f(x'),
\end{equation}
where $\mathcal{C}$ is the Carrollian manifold and $\epsilon$ is the volume form on the Carrollian manifold. The Fourier transform for the Green's function is given by 
\begin{equation}
    G(x-x') = \int_{\mathbb{R}^{n+1}} \tfrac{d^{d+1}\boldsymbol{k}}{(2\pi)^{d+1}}e^{i(\omega (t-t') - \vec{k}\cdot (\vec{x}-\vec{x'}))}\title{G}(\boldsymbol{k}),
\end{equation}
where $\boldsymbol{k} = (\omega,\vec{k})$ is the Carrollian version of momentum. Substituting into \eqref{green} we get
\begin{equation}
    (g_1\omega^2 + k^2)(g_2\omega^2 + k^2)\cdots (g_n\omega^2 + k^2)\tilde{G}(\boldsymbol{k}) = i,
\end{equation}
where $k= |\vec{k}|$. This is an algebraic equation which can be solved by
\begin{equation}
    \tilde{G}(\boldsymbol{k}) = \tfrac{-i}{(g_1\omega^2 + k^2)(g_2\omega^2 + k^2)\cdots (g_n\omega^2 + k^2)}.
\end{equation}
The poles of the Fourier transformed Green's function are at \begin{equation}
    \omega = \pm \tfrac{1}{\sqrt{- g_i}} k,
\end{equation}
 meaning that there are $n$ coexisting massless particles moving at different speeds as explained above.

 We conclude this section by showing that, although the most general Carrollian higher-derivative scalar field Lagrangian can not be derived from the limit of its Lorentzian counterpart, we can still derive all the diffeomorphism invariant terms (non-conformal) from various orders of the Carrollian expansion of the Lorentzian Lagrangian \eqref{rel lag}. In flat spacetime in adapted coordinates, $\Box$ can be written as
\begin{equation}
\Box = -\tfrac{1}{c^2}\partial^2_t +\partial_a\partial^a,
\end{equation}
thus,
\begin{equation}
    \Box^n = \sum_{k=0}^n \tfrac{n!}{k!(n-k)!} (\tfrac{-1}{c^2})^{n-k}\partial_t^{2k}(\partial_a\partial^a)^{n-k}.
\end{equation}
 Substituting in the Lagrangian we get
\begin{equation}\label{flat space Lagrangian}
    \mathcal{L}=  \sum_{k=0}^n \tfrac{n!}{k!(n-k)!} (\tfrac{-1}{c^2})^{k} \mathcal{L}_{2k},
\end{equation}
where $\mathcal{L}_{2k} = \phi\partial_t^{2k}(\partial_a\partial^a)^{n-k} \phi$ which matches with \eqref{lagr22}. This means that although the most general Carrollian Lagrangian can not be derived from the expansion of the Lorentzian Lagrangian, and different orders in the expansion can not coexist, all Carroll invariant quantities (with n derivatives) can be obtained individually as different orders of the Carroll expansion of the Lorentzian Lagrangian.

\section{Ostrogradsky's theorem for Carrollian scalar fields}\label{sec6}

In this section, we rewrite Ostrogradsky's theorem but for Carrollian field theories. We derive a general condition for a higher derivative Carrollian scalar field theory in order not to have Ostrogradsky's instability and show that it can be satisfied neither for Lorentzian theories nor for Carrollian theories derived as a limit of Lorentzian theories.

We begin with the Lagrangian
\begin{equation}
    \mathcal{L} = \phi \Box_{g_1}\Box_{g_2}\cdots\Box_{g_n}\phi,
\end{equation}
and calculate the canonical coordinates \eqref{can mom}:
\begin{equation}\label{eq.}
    \begin{aligned}
        Q_i = \partial_t^{i-1}\phi.
    \end{aligned}
\end{equation}
The naive choice for canonical momenta can be calculated from \eqref{can mom} as
\begin{equation}
    P_i = 2 (-1)^{n-i}g_1g_2\cdots g_n \partial_t^{2n-i}\phi,
\end{equation}
however, this is not a conjugate to $Q_i$ when one of the coupling constants equals zero since they do not satisfy the canonical commutation relations. Thus, we have to perform a rescaling to get the true canonical momenta. The reader can refer to the appendices of ref. \cite{Koutrolikos:2023evq} for more details on the validity of canonical momenta in Carrollian and Galilean theories. These are defined using $P_i$ as
\begin{equation}
    P_i = 2 (-1)^{n-i}g_1g_2\cdots g_n p_i,
\end{equation}
where the parameter independent momenta $p_i = \partial_t^{2n-i}\phi$ with $i \in \{1,2,\cdots,n\}$.
The Hamiltonian is then
\begin{equation}
    \mathcal{H} = \sum_{i=0}^{n}2(-1)^{n-i-1} g_1g_2\cdots g_n p_{i+1}\phi^{(i)}-\mathcal{L}.
\end{equation}
As we seen before, we can write $\phi^{(n)} = f(Q_1,\cdots,Q_{n},P_n)$, substituting in the Hamiltonian, we get
\begin{equation}
    \mathcal{H} = \sum_{i=1}^{n-1}2(-1)^{n-i} g_1g_2\cdots g_n p_{i}Q_{i} + 2 g_1g_2\cdots g_n p_{n}f(Q_1,\cdots,Q_{n-1},P_n) -\mathcal{L},
\end{equation}
Using \eqref{eq.}, we can see that $\phi^{(n)}=f(Q_1,\cdots,Q_n,P_n)=p_n$. Substituting in the Hamiltonian, we get
\begin{equation}
    \mathcal{H} = \sum_{i=1}^{n-1}2(-1)^{n-i} g_1g_2\cdots g_n p_{i}Q_{i} + 2 g_1g_2\cdots g_n p_n^2 -\mathcal{L}.
\end{equation}

One can clearly see that the Hamiltonian has linear terms in the momenta (except $p_n$). However, unlike Lorentzian theories, we have here the freedom to choose the constants $g_1,\cdots,g_n$, and it is evident that if we choose at least one of them to be zero (which can be interpreted as a particle with infinite velocity or an instantaneous force similar to Newtonian gravity), all linear terms will not contribute to the Hamiltonian. This is still in line with the original Ostrogradsky's theorem, since if we equate one of the coupling constants to zero, the Lagrangian becomes degenerate. The advantages of Carrollian theories are, first of all, that the restriction on the theory is minimal compared to Lorentzian theories (setting a coupling constant to zero rather than restricting the phase space), and we still have a larger class of theories than Lorentzian theories. Also, removing Ostrogradsky's ghosts is done by a well defined general condition unlike Lorentzian where this condition does not exist, and theories must be examined individually. 

It is also evident that this condition cannot be satisfied by theories derived as a limit of Lorentzian theories, since in that case all constants are equal to $-1/c^2$ supporting the results in \cite{Chen:2012au} that Ostrogradsky's instability (in Lorentzian theories) cannot be removed by a choice of the parameters in the theory but rather by constraining the phase space. We can now write Ostrogradsky's theorem for Carrollian scalar field theories.

\begin{theorem}
  Any higher derivative Carrollian scalar field theory with Lagrangian of the form \eqref{newlag} with $n\geq 2$ leads to an unbounded Hamiltonian unless at least one constant satisfies $g_i = 0$.
\end{theorem}

\begin{corollary}
  All purely magnetic higher derivative Carrollian scalar field theories have bounded Hamiltonians.
\end{corollary}
\begin{proof}
 The proof is immediate from the previous theorem, since purely magnetic theories are defined by setting all constants $g_1,\cdots,g_n$ to zero.
\end{proof}

As an example, consider the complex higher derivative scalar field of Lagrangian
\begin{equation}\label{example lagr}
    \mathcal{L} = \phi \Box_{g_1}\Box_{g_2}\Box_{g_3}\phi^{\ast} - \lambda(\phi \phi^{\ast})^2.
\end{equation}
The coordinates of the phase space for $\phi$ are $Q_0=\phi, Q_2=\dot{\phi}, Q_3=\ddot{\phi}$ as well as the parameter-independent momenta $p_1=\phi^{(5)}$, $p_2 = \phi^{(4)}$, $p_3= \phi^{(3)}$. The Hamiltonian is then
\begin{equation}
    \mathcal{H} = 2g_1g_2g_3 p_1Q_1 + 2g_1g_2g_3 p_2Q_2 + 2g_1g_2g_3 p_1^2 + \textrm{c.c.}-\mathcal{L},
\end{equation}
where $\textrm{c.c.}$ stands for complex conjugate of the preceding terms. As one can see the first two terms (and their complex conjugates) are linear in $p_1$ and $p_2$, which results in an instability as follows: Consider a state $\ket{p}$ in the physical space of the theory with a definite energy, and promote the Hamiltonian into an operator, we have
\begin{equation}\label{shr}
    \mathcal{H}\ket{p} = E\ket{p},
\end{equation}
where $E$ is the eigenvalue of the Hamiltonian operator representing the energy of the state. It is easy to see that $E$ is unbounded due to the linearity of the momentum in the Hamiltonian. However, this would not have been a problem in the free theory (vacuum can not decay), but in our case (with interactions) vacuum can decay indefinitely into particles and antiparticles with unbounded energy i.e. Ostrogradsky's instability. This instability can be removed by setting $g_1$, $g_2$ or $g_3$ to zero. In this case the Hamiltonian would be
\begin{equation}
    \mathcal{H} = -\mathcal{L},
\end{equation}
and the eigenvalue of the Hamiltonian in equation \eqref{shr} is bounded by the interaction term i.e. $E_{\min} = (\lambda) \langle|\phi|^4\rangle = 0$. Thus, the minimum energy that can be attained is $0$ and the vacuum can not decay indefinitely.

\section{Conclusions}\label{sec7}
In this paper we defined the most generic non-conformal massless higher derivative scalar field theory with $2n$ derivatives intrinsically on a flat Carrollian manifold. The solutions for the field equations for this theory represent $n$ particles propagating with different speeds depending on the coupling constants in the Lagrangian. This is in contrast to scalar fields on Lorentzian manifolds where the solution is one type of particles moving at the speed of light. As plane waves this gives interference patterns which would be impossible on Lorentzian manifolds (since all waves move at the same speed on Lorentzian manifolds but it is not true on Carrollian ones). We then studied the stability of these theories. The main result is that Carrollian theories are more resilient to Ostrogradsky's instabilities than Lorentzian ones. We proved that Ostrogradsky's instabilities can be removed by choosing the coupling constants appropriately, something that is not viable in Lorentzian theories where constraints on the phase space must be imposed regardless of the values of the parameters in the Lagrangian. Some Carrollian theories (the pure magnetic ones) do not even suffer from these instabilities at all.

Since we mostly studied free scalar fields, particles do not interact. A possible future direction is to study the massive counterpart of the Carrollian field theories we studied. It is known that massive higher derivative scalar field theories incorporate particles of different masses. So, an interesting question is how particles with different masses and of different types (in the sense that they have a different value for a coupling constant in their dispersion relation) interact. A possible generalization is to consider all the previous scenarios on a curved background, and a more ambitious generalization to this paper is to consider different field theories: Dirac fields, Maxwell fields, gravity, etc.

\section*{Acknowledgements}

The authors would like to thank Luca Ciambelli (Peremiter institute, Canada) for stimulating discussions. P.T. and I.K. were supported by Primus grant PRIMUS/23/SCI/005 from Charles University. I.K. is also grateful to Charles University Research Center Grant No. UNCE24/SCI/016 for their support.

\end{document}